\documentclass{elsart}
\journal{J. of Zhejiang University SCIENCE}

\usepackage{amsfonts,bm}
\usepackage{graphicx}
\usepackage{amsmath,amssymb}
\usepackage{url,cite}

\usepackage{lineno}
\linenumbers

\graphicspath{{figures/}}

\newtheorem{proposition}{Proposition}
\newtheorem{corollary}{Corollary}
\newenvironment{proof}{\noindent\textit{Proof}: }{\hfill$\blacksquare$\vskip 0.5\baselineskip}

\begin{document}

\newlength\figwidth
\setlength\figwidth{0.48\columnwidth}

\begin{frontmatter}
\title{Cryptanalysis of an image encryption scheme based on the Hill cipher}
\author[hk-cityu]{Chengqing Li\corauthref{corr}}, \author[cn-zju]{Dan Zhang}, and \author[hk-cityu]{Guanrong Chen}
\corauth[corr]{Corresponding author: Chengqing Li
(swiftsheep@hotmail.com).}

\address[hk-cityu]{Department of Electronic Engineering, City University of Hong Kong,
Kowloon Tong, Hong Kong SAR, China}
\address[cn-zju]{College of Computer Science and Technology, Zhejiang University, Hangzhou 310027, Zhejiang, China}

\begin{abstract}
This paper studies the security of an image encryption scheme based
on the Hill cipher and reports its following problems: 1) there is a
simple necessary and sufficient condition that makes a number of
secret keys invalid; 2) it is insensitive to the change of the
secret key; 3) it is insensitive to the change of the plain-image;
4) it can be broken with only one known/chosen-plaintext; 5) it has
some other minor defects.
\end{abstract}
\begin{keyword}
cryptanalysis \sep encryption \sep Hill cipher \sep known-plaintext
attack.

CLC: TN918, TP393.08.
\end{keyword}
\end{frontmatter}

\section{Introduction}

The history of cryptography can be traced back to the secret
communication among people thousands of years ago. With the
development of human society and industrial technology, theories and
methods of cryptography have been changed and improved gradually,
and meanwhile cryptanalysis has also been developed. In 1949,
Shannon published his seminar paper ``Communication theory of
secrecy systems" \cite{Shannon:CommunicationTheory1949}, which
marked the beginning of the modern cryptology.

In the past two decades, the security of multimedia data has become
more and more important. However, it has been recognized that the
traditional text-encryption schemes cannot efficiently protect
multimedia data due to some special properties of the multimedia
data, such as strong redundancy and bulk size of the uncompressed
data. To meet this challenge, a number of special image encryption
schemes based on some nonlinear theories were proposed
\cite{Li:ChaosImageVideoEncryption:Handbook2004,Li:Dissertation2003,Lcq:MasterThesis2005}.
Yet, many of them are found to be insecure from the view point of
cryptography \cite{Li:AttackingCNN2004,Li:AttackTDCEA2005,
Li:AttackingISNN2005,Li:AttackDSEA2004,Li:AttackJEI2006,
Li:AttackingBitshiftXOR2006,Li:BasicRequire:IJBC06,
David:AttackingFilterBank2007,Li:AttackingIVC2007,Li:AttackingISWBE2006,Zhou:AnalysisHuffman2007,
LiShujun:PVEA:IEEETCASVT2007,AlvarezLi:CBM2007}.

In \cite{ISMAIL:JZJUS06}, Ismail et al. tried to encrypt images
efficiently by modifying the classical Hill cipher
\cite{Hill:HillCipher1929}. This paper studies the security of the
scheme proposed in \cite{ISMAIL:JZJUS06} and reports the following
findings: 1) there exist a number of invalid secret keys; 2) the
scheme is insensitive to the change of the secret key; 3) the scheme
is insensitive to the change of the plain-image; 4) the scheme can
be broken with only one known/chosen plain-image; 5) the scheme has
some other minor performance defects.

The rest of this paper is organized as follows. The next section
briefly introduces the encryption scheme to be studied.
Section~\ref{sec:Cryptanalysis} presents detailed cryptanalysis of
the scheme. The last section concludes the paper.

\section{The image encryption scheme to be studied}
\label{sec:scheme}

The scheme proposed in \cite{ISMAIL:JZJUS06} scans the gray scales
of a plain-image $\bm{P}$ (or one channel of a color image) of size
$M\times N$ in a raster order and divides it into $\lceil MN/m
\rceil$ vectors of size $m$: $\{\bm{P}_l\}_{l=1}^{\lceil MN/m
\rceil}$, where $\bm{P}_l=\{\bm{P}((l-1)\cdot m+1), \cdots,
\bm{P}((l-1)\cdot m+m)\}$ (the last vector is padded with some zero
bytes if $MN$ can not be divided by $m$). Then, the vectors
$\{\bm{P}_l\}_{l=1}^{\lceil MN/m \rceil}$ are encrypted in
increasing order with the following function:
\begin{equation}
\bm{C}_l=(\bm{P}_l\cdot \bm{K}_l)\bmod 256, \label{eq:encryption}
\end{equation}
where $\bm{K}_1=(\bm{K}_1[i,j])_{m\times m}$, $\bm{K}_1[i,j]\in
\mathbb{Z}_{256}$, the initial state of $\bm{K}_{l\ge 2}$ is set to
be $\bm{K}_{l-1}$, and then every row of $\bm{K}_l$ is generated
iteratively with the following function, for $i=1\sim m$:
\begin{equation}
\bm{K}_l[i,:]=(\bm{IV}\cdot\bm{K}_{l})\bmod 256,
\label{eq:Updatekey}
\end{equation}
where $\bm{IV}$ is a vector of size $1\times m$ and $\bm{IV}[i]\in
\mathbb{Z}_{256}$. Finally, the cipher-image is obtained as
$\bm{C}=\{\bm{C}_l\}_{l=1}^{\lceil MN/m \rceil}$.

The secret key of the encryption scheme includes three parts: $m$,
$\bm{K}_0$, and $\bm{IV}$.

The decryption procedure is the same as the above encryption
procedure except that Eq.~(\ref{eq:encryption}) is replaced by the
following function:
\begin{equation}
\bm{P}_l=(\bm{C}_l\cdot \bm{K}_l^{-1})\bmod 256,
\end{equation}
where  $(\bm{K}_l\cdot \bm{K}_l^{-1})\bmod 256=\bm{\mathrm{I}}$, the
identity matrix.

\section{Cryptanalysis}
\label{sec:Cryptanalysis}

\subsection{Some Defects of the Scheme}

\subsubsection{Invalid keys}

An invalid key is a key that fails to ensure the success of the
encryption scheme.

From the following Fact~\ref{fact:invertible} and
Corollary~\ref{coro:det}, one can see that one secret key in the
above-described scheme is invalid if and only if
$\gcd(\bm{K}_1,256)\neq 1$ or $\bm{IV}[i]\bmod 2=0$.

\begin{fact}
A matrix $\bm{K}$ is invertible in $\mathbb{Z}_n$ if and only if
$\gcd(\det(\bm{K}),n)=1$. \label{fact:invertible}
\end{fact}

\begin{proposition}
$\det(\bm{K}_l)=\left(\prod\limits_{i=1}^m\bm{IV}[i]\right)\det(\bm{K}_{l-1})$.
\label{prop:det}
\end{proposition}
\begin{proof}
According to Eq.~(\ref{eq:Updatekey}), there is a relation between
$\bm{K}_l$ and $\bm{K}_{l-1}$, as follows:
\begin{equation}
\bm{K}_l=\left(
\begin{array}{cc}
 & \sum\limits_{i=1}^m \bm{IV}[i]\bm{K}_{l-1}[i,:]\\
 & \bm{IV}[1]\bm{K}_l[1,:]+\sum\limits_{i=2}^m \bm{IV}[i] \bm{K}_{l-1}[i,:]\\
 &         \vdots    \\
 &\sum\limits_{i=1}^{m-1} \bm{IV}[i] \bm{K}_{l}[i,:]+\bm{IV}[m]
\bm{K}_{l-1}[m,:]
\end{array}
\right)\bmod 256.
\end{equation}

Subtracting $\sum\limits_{i=1}^{i_0-1} \bm{IV}[i] \bm{K}_{l}[i,:]$
from $\bm{K}_{l}[i_0,:]$ for $i_0=m\sim 2$, one gets

\begin{equation}
\bm{K}'_l=\left(
\begin{array}{cc}
 & \sum\limits_{i=1}^m \bm{IV}[i]\bm{K}_{l-1}[i,:]\\
 & \sum\limits_{i=2}^m \bm{IV}[i] \bm{K}_{l-1}[i,:]\\
 &         \vdots    \\
 & \bm{IV}[m]\bm{K}_{l-1}[m,:]
\end{array}
\right)\bmod 256.
\end{equation}

Subtracting $\bm{K}'_{l}[i_0,:]$ from $\bm{K}'_{l}[i_0-1,:]$ for
$i_0=2\sim m$, one has

\begin{equation}
\bm{K}''_l=\left(
\begin{array}{cc}
 & \bm{IV}[1]\bm{K}_{l-1}[1,:]\\
 & \bm{IV}[2] \bm{K}_{l-1}[2,:]\\
 &         \vdots    \\
 & \bm{IV}[m]\bm{K}_{l-1}[m,:]
\end{array}
\right)\bmod 256.
\end{equation}

Obviously, $\det(\bm{K}_l)=\det(\bm{K}'_l)
=\det(\bm{K}''_l)=\left(\prod\limits_{i=1}^m\bm{IV}[i]\right)
\det(\bm{K}_{l-1})$, which completes the proof of the proposition.
\end{proof}

\begin{corollary}
$\det(\bm{K}_l)=\left(\prod\limits_{i=1}^m\bm{IV}[i]\right)^{l-1}\det(\bm{K}_{1})$.
\label{coro:det}
\end{corollary}
\begin{proof}
The result directly follows from Proposition~\ref{prop:det}.
\end{proof}

\subsubsection{Insensitivity to the change of the secret key}

Although it is claimed in \cite[Sec. 5]{ISMAIL:JZJUS06} that the
encryption scheme is very sensitive to the change of the sub-keys
$\bm{K}_1$, $\bm{IV}$, this is not true.

Let's first study the influence on $\bm{K}_{l\ge 2}$ if only one bit
of $\bm{K}_1$ is changed. Without loss of generality, assume that
the $n$-th significant bit of $\bm{K}_1(1,j_0)$ is changed from zero
to one, where $0\le n\le 7$. Let $\widetilde{\bm{K}}_l$ denote the
modified version of $\bm{K}_l$. The change
$\bm{D}_l=\widetilde{\bm{K}}_l-\bm{K}_l$ can be presented by the
following two equations:

\begin{equation}
\bm{D}_l[:,j]\equiv0, \mbox{for } j\neq j_0,
\end{equation}
\begin{equation}
\bm{D}_l[:,j_0]=\left(
\begin{array}{cc}
 & \sum\limits_{i=1}^m\bm{IV}[i]\bm{D}_{l-1}[i, j_0]\\
 & \bm{IV}[1]\bm{D}_{l}[1, j_0]+\sum\limits_{i=2}^m\bm{IV}[i]\bm{D}_{l-1}[i, j_0]\\
 &         \vdots    \\
 & \sum\limits_{i=1}^{m-1}\bm{IV}[i]\bm{D}_{l}[i, j_0]+\bm{IV}[m]\bm{D}_{l-1}[m, j_0]
\end{array}
\right)\bmod 256,
\label{eq:Difference1}
\end{equation}
where $\bm{D}_1[1,j_0]=2^n$, $\bm{D}_1[i,j_0]=0$, $i=2\sim m$.

Since $\bm{IV}[i]\bmod 2\neq 0$, $\bm{D}_l[i,j_0]\neq 0$ always
exist. From Eq.~(\ref{eq:Difference1}), one can see that
$\bm{D}_l[i,j_0]\ge 2^n$ exists, which means that only the $n_0$-th
bit of $\bm{C}_l[j_0]$ may possibly be changed, where $n_0\ge n$.
Note also that there is no influence on $\bm{C}_l$ if
$$(\bm{P}_l\bm{D}_l[:,j_0])\bmod 256=0.$$

To verify the above analysis, an experiment has been carried out
using a plain-image ``Lenna" with the secret key
\begin{equation}
m=4, \bm{IV}=(
\begin{matrix}
3 & 9 & 17 & 33
\end{matrix}),
\bm{K}_1=\left(
\begin{matrix}
11 & 2 & 3 & 7 \\
8 & 5 & 19 & 103 \\
201 & 203 & 119 & 150 \\
7 & 9 & 21 & 35 \\
\end{matrix}
\right). \label{eq:secretkey}
\end{equation}
Only the $5$-th significant bit of $\bm{K}_1[1,2]$ is changed,
namely $\widetilde{\bm{K}}_1[1,2]=(\bm{K}_1[1,2]+2^{5})\bmod 256$.
Let $\widetilde{\bm{C}}$ denote the cipher-image corresponding to
$\widetilde{\bm{K}}_1$. The bit-planes of difference
$|\widetilde{\bm{C}}-\bm{C}|$ are shown in
Fig.~\ref{figure:InsensitivtyK1}, which demonstrates the very weak
sensitivity of the encryption scheme with respect to $\bm{K}_1$.

\begin{figure}[!htb]
\centering
\begin{minipage}[t]{0.5\figwidth}
\includegraphics[width=0.5\figwidth]{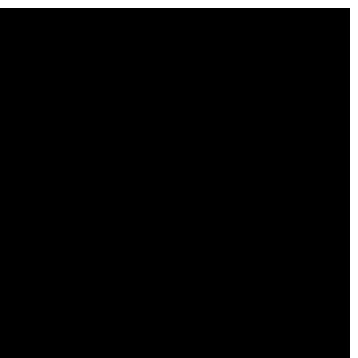}
a) $0\sim 4$-th
\end{minipage}
\begin{minipage}[t]{0.5\figwidth}
\includegraphics[width=0.5\figwidth]{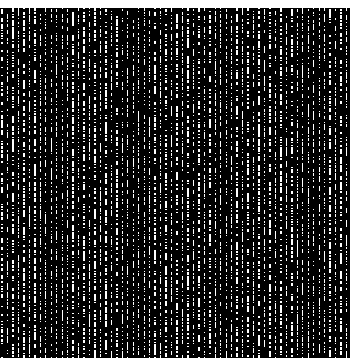}
b) $5$-th
\end{minipage}
\begin{minipage}[t]{0.5\figwidth}
\includegraphics[width=0.5\figwidth]{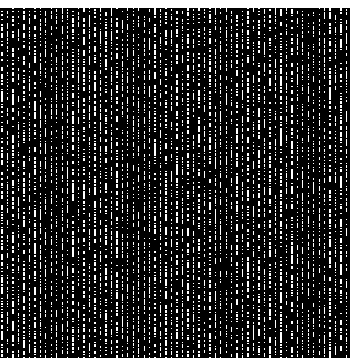}
c) $6$-th
\end{minipage}
\begin{minipage}[t]{0.5\figwidth}
\includegraphics[width=0.5\figwidth]{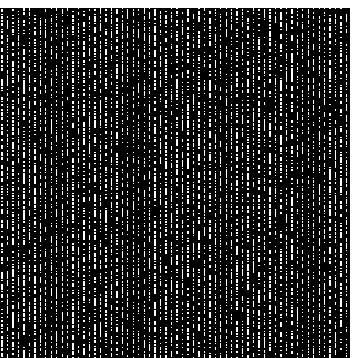}
d) $7$-th
\end{minipage}
\caption{The bit-planes of $|\widetilde{\bm{C}}-\bm{C}|$ when one
bit of $\bm{K}_1$ is changed.} \label{figure:InsensitivtyK1}
\end{figure}

Now, consider the influence on $\bm{K}_{l\ge 2}$ if only one bit of
$\bm{IV}$ is changed. Without loss of generality, assume the $n$-th
significant bit of $\bm{IV}[1]$ is changed from zero to one.
Similarly, let $\bm{D}_l$ denote the change of $\bm{K}_l$. Due to
the extremely complex formulation of $\bm{D}_{l\ge 3}$, only
$\bm{D}_2$ is shown here.

\begin{equation}
\bm{D}_2[:,j]=\left(
\begin{array}{cc}
 & \bm{K}_1[1,j]2^n \\
 & \bm{D}_2[1,j](\bm{IV}[1]+2^{n})+\bm{K}_2[1,j]2^n\\
 & \bm{D}_2[2,j]+\bm{IV}[2]\bm{D}_2[2,j]\\
 &         \vdots    \\
 & \bm{D}_2[2,j]+\sum\limits_{i=2}^{m-1}\bm{IV}[i]\bm{D}_2[i,j]
\end{array}
\right)\bmod 256,\label{eq:Difference2}
\end{equation}
where $j=1\sim m$.

To see the influence of the change of $\bm{IV}$, an experiment has
been carried out using plain-image ``Lenna", with the same secret
key shown in Eq.~(\ref{eq:secretkey}) above. Only the $5$-th
significant bit of $\bm{IV}[1]$ is changed, namely
$\widetilde{\bm{IV}}[1]=(\bm{IV}[1]+2^{5})\bmod 256$. The bit-planes
of difference between cipher-images corresponding to $\bm{IV}$ and
$\widetilde{\bm{IV}}$, respectively, are shown in
Fig.~\ref{figure:InsensitivtyIV}.

Comparing Fig.~\ref{figure:InsensitivtyK1} and
Fig.~\ref{figure:InsensitivtyIV}, one can see that the sensitivty
with respect to $\bm{IV}$ is much stronger than the one with respect
to $\bm{K}_1$, which agrees with the above theoretical analysis. But
one bit change of a sub-key of a secure cipher should cause every
bit of the ciphertext changed with a probability of $\frac{1}{2}$.
Obviously, the sensitivity of the encryption scheme with respect to
sub-keys $\bm{K}_1$, $\bm{IV}$ is very far from this requirement.

\begin{figure}[!htb]
\centering
\begin{minipage}[t]{0.5\figwidth}
\includegraphics[width=0.5\figwidth]{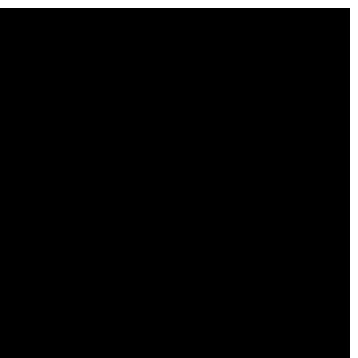}
a) $0\sim 4$-th
\end{minipage}
\begin{minipage}[t]{0.5\figwidth}
\includegraphics[width=0.5\figwidth]{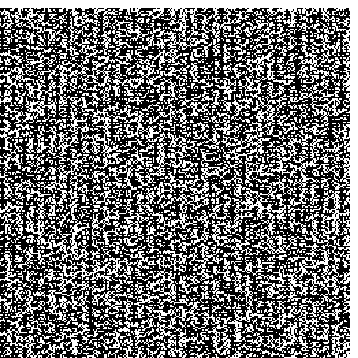}
b) $5$-th
\end{minipage}
\begin{minipage}[t]{0.5\figwidth}
\includegraphics[width=0.5\figwidth]{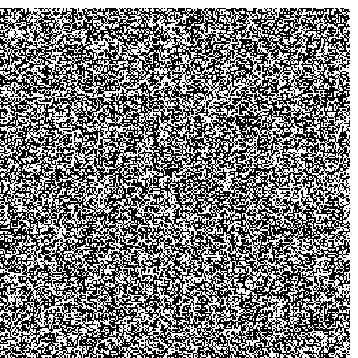}
c) $6$-th
\end{minipage}
\begin{minipage}[t]{0.5\figwidth}
\includegraphics[width=0.5\figwidth]{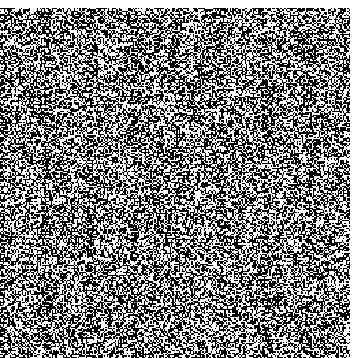}
d) $7$-th
\end{minipage}
\caption{The bit-planes of $|\widetilde{\bm{C}}-\bm{C}|$ when one
bit of $\bm{IV}$ is changed.} \label{figure:InsensitivtyIV}
\end{figure}

\subsubsection{Insensitivity to the change of the plain-image}

This property is especially important for image encryption since an
image and its watermarked version may be encrypted simultaneously.

Since the role of $\bm{P}_l$ in Eq.~(\ref{eq:encryption}) is exactly
the same as that of $\bm{IV}$ in Eq.~(\ref{eq:Updatekey}), the
analysis about its insensitivity to the change of the plain-image
can be carried out just like the case about the sub-key $\bm{IV}$
discussed above.

\subsubsection{Some other problems}

The encryption scheme has the following additional problems:

\begin{enumerate}
\item \textit{cannot encrypt plain-image of a fixed value zero};

\item \textit{efficiency of implementation is low}:
From\cite[Thorem 2.3.3]{Jeffrey:Cryptologia05}, one can see that the
number of invertible matrices of size $m\times m$ in
$\mathbb{Z}_{256}$ is
\begin{equation}
  |GL(m,\mathbb{Z}_{256})|=2^{7m^2}\prod_{k=0}^{m-1}(2^m-2^k).
\end{equation}
Thus, the probability that a matrix of size $m\times m$ in
$\mathbb{Z}_{256}$ is invertible is
\begin{equation}
p_m=\frac{2^{7m^2}\prod_{k=0}^{m-1}(2^m-2^k)}{2^{8m^2}}=\prod_{k=1}^{m}(1-2^{-k})\approx\frac{1}{3}.
\end{equation}
So, it needs $O(3m^2)$ and $O(m^2\cdot MN)$ times of computations,
respectively, for checking the reversibility of $\bm{K}_1$ and for
calculating $\{\bm{K}^{-1}_l\}_{l=1}^{\lceil MN/m\rceil}$. Note that
these computations have no direct contributions to protecting the
plain-image.

\item \textit{the scope of sub-key $m$ is limited}:
As discussed above, the larger the value $m$ the higher the
computational cost.

\item \textit{the confusion capability is weak}:
This problem is caused by the linearity of the main encryption
function. To demonstrate this defect, the encryption result of one
special plain-image is shown in Fig.~\ref{figure:SpecialImage},
where Figure~\ref{figure:SpecialImage}b) also effectively disproves
the conclusion about the quality of encryption results given in
\cite[Sec. 4]{ISMAIL:JZJUS06}.
\end{enumerate}

\begin{figure}[!htb]
\centering
\begin{minipage}[t]{0.8\figwidth}
\centering
\includegraphics[width=0.8\figwidth]{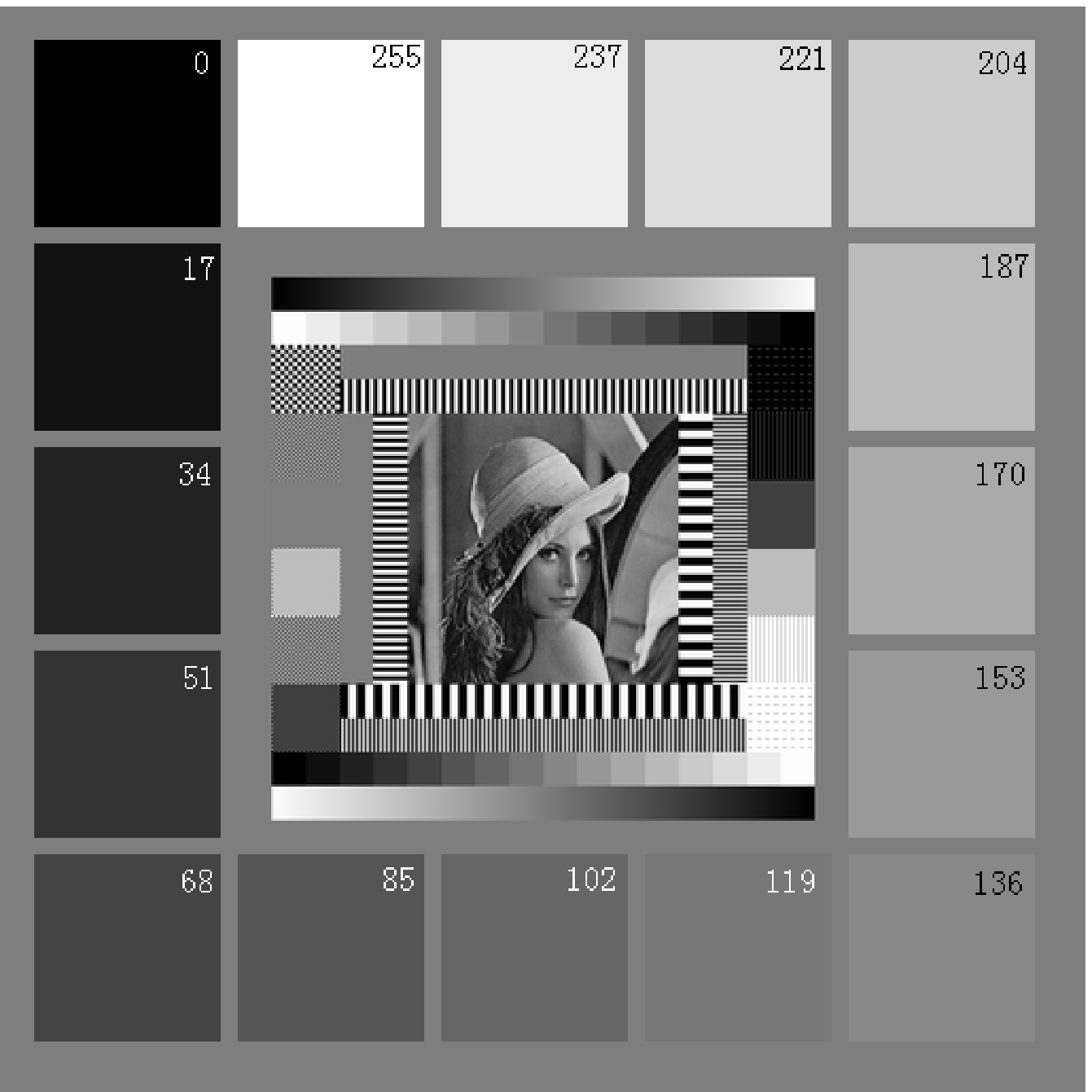}
a) plain-image
\end{minipage}
\begin{minipage}[t]{0.8\figwidth}
\centering
\includegraphics[width=0.8\figwidth]{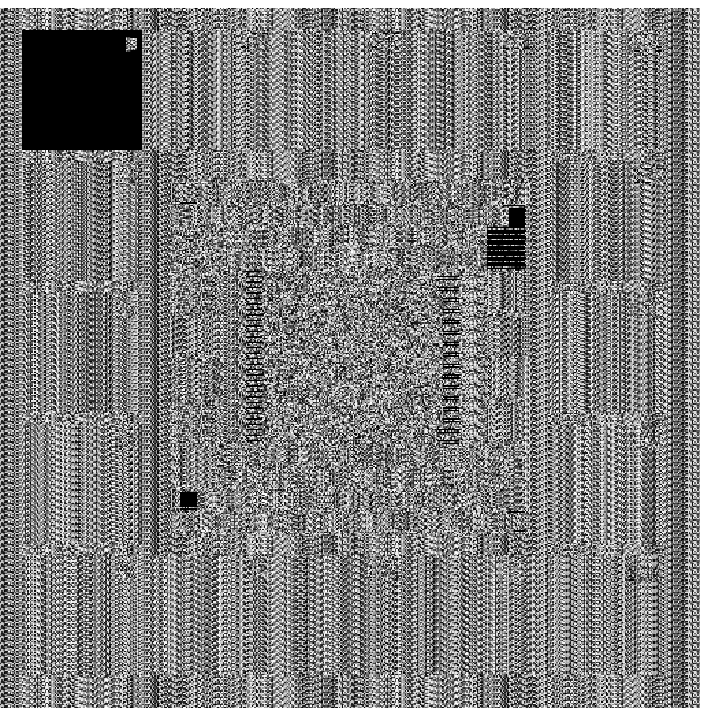}
b) cipher-image
\end{minipage}
\caption{A special test image, ``Test\_pattern".}
\label{figure:SpecialImage}
\end{figure}

\subsection{Known/Chosen-Plaintext Attack}

The known/chosen-plaintext attack works by reconstructing the secret
key or its equivalent based on some known/chosen plaintexts and
their corresponding ciphertexts.

For this encryption scheme, the equivalent key
$\{\bm{K}_l\}_{l=1}^{\lceil MN/m \rceil}$ can be reconstructed from
$m$ plain-images $\bm{P}^{(1)}\sim \bm{P}^{(m)}$ and their
corresponding cipher-images $\bm{C}^{(1)}\sim \bm{C}^{(m)}$ by using
\begin{equation} \bm{K}_l= \left(\bm{P}^{(B)}_l\cdot\left(
\begin{array}{cc}
 & \bm{C}^{(1)}_l\\
 & \bm{C}^{(2)}_l\\
 &         \vdots    \\
 & \bm{C}^{(m)}_l
\end{array}
\right)\right)\bmod 256,\label{eq:knownplainimage}
\end{equation}
where
\begin{equation}
\bm{P}^{(B)}_l= \left(
\begin{array}{cc}
 & \bm{P}^{(1)}_l\\
 & \bm{P}^{(2)}_l\\
 &                 \vdots    \\
 & \bm{P}^{(m)}_l
\end{array}
\right)^{-1}.
\end{equation}
The reversibility of $\bm{P}^{(B)}_l$ can be ensured by utilizing
more than $m$ plain-images or by choosing $m$ special plain-images.
Note that the above known/chosen-plaintext attack can be carried out
with only one know/chosen plain-image due to the very short period
of sequence $\{\bm{K}_l[:,j]\}_{l=1}^{\lceil MN/m\rceil}$ for
$j=1\sim m$. To study the period of this sequence, 10,000 tests have
been done for a given value of $\bm{IV}$ of size $1\times 3$, where
$\bm{K}_1$ is selected randomly. The numbers of tests where the
corresponding sequence $\{\bm{K}_l(:,1)\}_{l=1}^{\lceil MN/m
\rceil}$ has period $p$, $N_p$, with some values of $\bm{IV}$, is
shown in Table~\ref{table:nonuniformityK}, which shows that the
period of $\{\bm{K}_l[:,j]\}_{l=1}^{\lceil MN/m\rceil}$ is indeed
very short.

\begin{table}[htbp]
\center\caption{Values of $N_p$ with some values of $\bm{IV}$,
$p=2^s$, $s=3\sim 9$.}
\begin{tabular}{*{7}{c|}c}
\hline $\bm{IV}$       &  $N_8$ & $N_{16}$  & $N_{32}$ & $N_{64}$ & $N_{128}$ & $N_{256}$ & $N_{512}$\\
\hline (91, 63, 45)  &  0     &  0        & 0        & 0        &  0        & 1463      & 8537     \\
\hline (113, 25, 219)&  14    &  34       & 127      & 561      & 3561      & 5703      & 0        \\
\hline (253, 115, 17)&  6     &  20       & 72       & 284      & 1081      & 8537      & 0        \\
\hline (1, 3, 5)     &  0     &  0        & 98       & 284      & 1081      & 8537      & 0        \\
\hline (5, 121, 247) &  7     &  36       & 132      & 561      &  3561     & 5703      & 0        \\
\hline
\end{tabular}
\label{table:nonuniformityK}
\end{table}

\section{Conclusion}

In this paper, the security and performance of an image encryption
scheme based on the Hill cipher have been analyzed in detail. It has
been found that the scheme can be broken with only one known/chosen
plain-image. There is a simple necessary and sufficient condition
that makes a number of secret keys invalid. In addition, the scheme
is insensitive to the change of the secret key/plain-image. Some
other performance defects have also been found. In conclusion, the
encryption scheme under study actually has much weaker security than
the original Hill cipher, therefore is not recommended for
applications.

\section{Acknowledgement}

This research was supported by the City University of Hong Kong
under the SRG project 7002134.

\bibliographystyle{elsart-num}
\bibliography{hill}
\end{document}